\documentclass[copyright,creativecommons]{eptcs}
\providecommand{\event}{CCA 2010} 
\usepackage{breakurl}
\usepackage{amssymb}
\usepackage[mathscr]{eucal}
\usepackage{amsmath}
\usepackage{amsfonts}

\usepackage{mathrsfs}
\usepackage{stmaryrd}
\usepackage{amsmath}
\usepackage{amsfonts}
\usepackage{amssymb}
\usepackage{amssymb, indentfirst, amsmath, indentfirst, fancyhdr}

\usepackage[mathscr]{eucal}
\usepackage{amsmath}

\newtheorem{theorem}{Theorem}[section]
\rm

\newtheorem{proposition}[theorem]{Proposition}

\newenvironment{proof}{\noindent\bf Proof.
\rm}{\mbox{}\hfill $\square$\vspace*{3mm}}

\def\IR{\mathbb{R}}

\begin{document}

\title{Computing the Solutions of the \\Combined Korteweg-de Vries  Equation by Turing Machines
\thanks{This work is supported by DFG  (446 CHV 113/240/0-1) and NSFC (10420130638)}
}
\author{Dianchen Lu,\,\,Qingyan Wang and  Rui Zheng
\institute{Nonlinear Scientific Research Center\\  Jiangsu University\\ Zhenjiang, 212013, P.R.China.}
\email{dclu@ujs.edu.cn}  
}
\def\titlerunning{Korteweg-de Vries  Equation}
\def\authorrunning{D.~Lu, Q.~Wang \& R.~Zheng}

\maketitle
\begin{abstract}
In this paper, we study the computability of the initial value problem of Combined Korteweg-de Vries (CKdV, for short) equation: $ u_t + uu_x + u^2u_x + u_{xxx} = 0$, $u(x, 0) = \phi(x)$. It is shown that, for any integer $s\geqslant 3$, the nonlinear solution operator  $K_R:H^s(\mathbb{R})\to C(\IR, H^s(\IR))$ which maps an initial condition data $\phi$ to the solution of the Combined KdV equation can be computed by a Turing machine.

\end{abstract}


\section{Introduction}
Differential equations are very popular mathematical models of real world problems. Not every differential equation has a well behaved solution. For those equations whose  well behaved solutions exist, we are interested in how they can be computed. Thus, the computability of the solution operators for different types of nonlinear differential equations becomes one of the most exciting topics in effective analysis. This answers questions of the type: is it possible to calculate the solutions of some real word problems algorithmically? The answers to these questions are unfortunately not always positive. However, there are a lot of very interesting equations whose solutions do exist and can be calculated. These equations can be called computably solvable equations, in other words, their solution operators are computable. This means that, there are Turing machines which can transfer the initial data to the solutions of the equation in some particular spaces. For example, Klaus Weihrauch and Ning Zhong \cite{dMS03} have shown that the initial value problem of Korteweg-de Vries (KdV) equation posed on the real line $\IR$: $u_t + uu_x + u_{xxx} = 0$, $t, x \in \IR$, $u(x, 0) = \phi(x)$ has a computable solution operator.

In this paper, we investigate a variation of Korteweg-de Vries equation:   $ u_t + uu_x + u^2u_x + u_{xxx} = 0$. This is often called Combined Korteweg-de Vries (CKdV) equation. The Combined KdV equation is also an important equation which is frequently used as a mathematical model in physics, hydrodynamics, biological and chemical fields. We will show that the solution operator of the CKdV equation is also computable. This extends the results of \cite{CL05,dMS03}. The proof of the main theorem is given in Section 2.

\section{Main result}\label{sec-main}

In this section, we use a similar approach as that in \cite{dMS03} and retrieve two estimates to prove the main result. We use Type 2 theory of effectivity (TTE) as computation model. More relevant details can be found in \cite{dMS03}.

For rigorous notation we occasionally write $u(x, t):= u(t)(x)$, where $u(t) :\IR \to
C(\IR;H^s(\IR))$. More precisely, we are interested in the following  initial value problem (IVP, for
short) of CKdV equation on the real line $\IR$,
\begin{eqnarray}\label{equ-3}
\left\{
\begin {array}{ll}u_t+uu_x+ u^2u_x + u_{xxx}=0,
 \quad\left(t, x\in
\mathbb{R}\right)\\ u\left({x, 0}\right) =\varphi\left(x\right)\in
H^s(\mathbb{R}).
\end{array} \right.
\end{eqnarray}

Then we consider the solution operator $K_\IR$ which maps the initial data $\varphi\in H^s(\mathbb{R})$ to the solution $u \in C(\IR;H^s(\IR))$, for $s\ge 3$.

The following is the equivalent integral equation of the initial
value problem (\ref{equ-3})
\begin{eqnarray}\label{equ-4}
u\left(t\right)=\mathcal {F}^{-1}\left({E\left(t\right)\cdot \mathcal
{F}\left(\varphi\right)}\right)-\int_0^t{\mathcal
{F}^{-1}\left({E\left({t-\tau}\right)\cdot \mathcal
{F}\left({\textstyle{d\over{dx}}\left({1\over2}{u^2\left(\tau\right)}+{1\over3}{u^3\left(\tau\right)}
\right)}\right)}\right)}\mathrm{d}\tau
\end{eqnarray}
where $u\left(t\right)\left(x\right):=u\left({x, t} \right),
E\left(t\right)\left(x\right) :=e^{ix^3t}$,  and  $\mathcal {F}
\left(\varphi\right)\left(x\right) =\frac{1}{\sqrt{2\pi} }
\int\limits_{\scriptstyle\rm R} {e^{- ix\xi}} \varphi \left(\xi\right)
\mathrm{d}\xi$.

We use the following iterative sequence with the initial data
$\varphi$ as the seed:
\begin{eqnarray}\label{equ-5}
\left\{\begin{array}{ll}
 v_0\left(t\right)
 =\mathcal {F}^{-1}\left({E\left(t\right)\cdot \mathcal {F}\left(\varphi\right)}\right)\\
 v_{j+1}\left(t\right)
 =v_0\left(t\right)-\int_0^t{\mathcal
{F}^{-1}\left({E\left({t-\tau}\right)\cdot \mathcal
{F}\left({\textstyle{d\over{dx}}\left({1\over2}{v_j^2\left(\tau\right)}+{1\over3}{v_j^3\left(\tau\right)}
\right)}\right)}\right)}\mathrm{d}\tau.
 \end{array} \right.
\end{eqnarray}

The iterative sequence (\ref{equ-5}) is contracting near $t=0, $
thus the sequence converges to a unique limit. Since the limit
satisfies the integral equation (\ref{equ-4}), it is the solution of
the initial value problem (\ref{equ-3}) near $t=0.$ To prove that
the solution operator is computable, we need to construct a type-2
Turing machine to compute it.

Firstly,  we define the operator:
\begin{eqnarray*}
 S( {u, \varphi, t})= \mathcal{F}^{-1}({E(t)\cdot \mathcal{F}(\varphi )})
 -\int_0^t{\mathcal
{F}^{-1}\left({E\left({t-\tau}\right)\cdot \mathcal
{F}\left({\textstyle{d\over{dx}}\left({1\over2}{u^2\left(\tau\right)}+{1\over3}{u^3\left(\tau\right)}
\right)}\right)}\right)}\mathrm{d}\tau
\end{eqnarray*}
which is (${[{\rho\to\delta_s}],  \delta_s, \rho, \delta_s}$)-computable. This
follows from Lemma 3.2 in \cite{dMS03} straightforwardly. Therefore,  the
function $\bar{S}({u, \varphi})(t):=S( {u, \varphi, t})$ is
$({[{\rho\to\delta_s} ], \delta_s, [{\rho\to\delta_s}]})$-computable. Then we
define the function $v: S (\IR) \times \mathbb{N}\to C ({\IR: S (\IR
)})$ by
\begin{equation*}
\begin{array}{l}
 v({\varphi, 0})    =\bar{S}({0, \varphi})\\
 v({\varphi, j+1})  =\bar{S}(v({\varphi, j}),\varphi).
\end{array}
\end{equation*}
It is easy to verify that $v$ is $({\delta_s,  \gamma_{\mathbb{N}},
[{\rho\to\delta_s} ]})$-computable.

Now we can show several propositions which lead to the proof of our main theorem.

\begin{proposition}\label{prop-3.1}
If $u\left({x, t}\right)$ is the solution of the $IVP$ (\ref{equ-3}), then
there is a computable function $e:{\mathbb{N}}\times{\mathbb{R}}
\times{\mathbb{R}} \to {\mathbb{R}}$ which is non-decreasing in the second and
third argument such that
\begin{equation*}
 \sup_{0\leqslant t\leqslant {T}}  \parallel u\left({x, t}\right)\parallel_s
 \leqslant e_{{T}}^s \left({\parallel\varphi \parallel_s }\right),
\end{equation*}
where $e_{{T}}^s\left(r\right):=e\left({s, {T}, r}\right)$,  $s$ is
an integer and $s\geqslant 3$.
\end{proposition}
\begin{proof}
See [3]
\end{proof}

\begin{proposition}\label{prop-2.2}
Let $v^0:= \bar{S}(0, \varphi)$,  and $v^{j+1}:=\bar{S}(v^j,
\varphi)$. If\begin{equation*}\begin{array}{l}\displaystyle~~~\alpha _T^s T^{1/2}[8(3 + T)^{3/2}\parallel\varphi
\parallel _s^2+4(3+T)\parallel
\varphi\parallel_s]\leqslant\frac{1}{2}\end{array},\end{equation*}
then we have
\begin{equation*}
\parallel v^{j+1}\left(t\right)-v^j\left(t\right)\parallel_s \leqslant
2^{-j}\left({3+T}\right)^{1/2}\parallel\varphi\parallel_s,
\end{equation*}
where $\alpha_T^s=\sqrt
s\cdot2^s\cdot T^{1/2}+1~.$
\end{proposition}
\begin{proof}

According to Proposition 2.1 and lemmas4.8 in [7], if $T>0~, u\in X_T^s$, we can obtain that
\begin{equation*}\int_0^T{\left\|{uu_x+u^2u_x}\right\|}_s dt \leqslant
\alpha _T^s T^{1/2}\left\| u
\right\|_{X_T^s}\left\|u\right\|_{X_T^s},\end{equation*} where
$\alpha_T^s=\sqrt
s\cdot2^s\cdot T^{1/2}+1,~$$X_T^s=\left\{{u\in C\left({\left[ {0, T}
\right];H^s\left({\mathbb{R}} \right);\Lambda _T^s\left(u\right)<
\infty }\right)}\right\}$ is a Banach space with the norm
$\left\|u\right\|_{X_T^s }$(detail see Definition 4.7 in [7]).

Then let $\displaystyle
W\left(t\right)\varphi=\frac{1}{\sqrt{2\pi}}\int_{
\mathbb{R}}e^{ix\xi}e^{i\xi^3t}\hat{\varphi}(\xi)\mathrm{d}\xi. $

Since  $v^0:=\bar{S}\left(0, \varphi\right),v^{j+1}:
=\bar{S}\left({v^j, \varphi}\right)$,
 by Lemma 4.9 and 4.10 in [7], for $j\geqslant 1, $
\begin{equation*}\begin{array}{l}\displaystyle~~~\parallel
v^j\parallel_{X_T^s }=\parallel W\left(t\right)\varphi-
\frac{1}{2}\int_0^t {W\left( {t - \tau }\right)}  {\left[
{(v^{j-1})^2}\right]_x}d\tau-
\frac{1}{3}\int_0^t {W\left( {t - \tau }\right)}  {\left[
{(v^{j-1})^3}\right]_x}d\tau
\parallel_{X_T^s}\vspace{\smallskipamount}\\
\displaystyle~~~~~~\leqslant \left( {3 + T} \right)^{1 / 2}\parallel
\varphi\parallel _s + \left( {3 + T} \right)^{1 / 2}\int_0^T
{\parallel v^{j - 1}v_x^{j - 1}+ \left( {v^{j - 1}} \right)} ^2v_x^{j - 1}\parallel _s
d\tau\vspace{\smallskipamount}\\
\displaystyle~~~~~~\leqslant \left( {3 + T} \right)^{1 / 2}\parallel \varphi
\parallel _s + \left( {3 + T} \right)^{1 / 2}\alpha _T^s T^{1 /
2}\parallel v^{j - 1}\parallel _{X_T^s }^2+ \left( {3 + T} \right)^{1 / 2}\alpha _T^s T^{1 /
2}\parallel v^{j - 1}\parallel _{X_T^s }^3.
\end{array}\end{equation*}\\
Let $T>0~, $ such that $\alpha _T^s T^{1/2}[8(3 + T)^{3/2}\parallel\varphi
\parallel _s^2+4(3+T)\parallel
\varphi\parallel_s]\leqslant\frac{1}{2}.$
From $\parallel v^0\parallel _{X_T^s }=
\parallel W\left( t \right)\varphi\parallel _{X_T^s } \le \left( {3 +
T} \right)^{1 / 2}\parallel \varphi\parallel _s$ we obtain by induction
\begin{equation*}\parallel v^j\parallel_{X_T^s }\leqslant
2\left({3+T}\right)^{1 / 2}\parallel \varphi\parallel_s .\left(for~
all~ j \in {\mathbb{N}} \right)\end{equation*} For $j\geqslant 2$,
\begin{equation*}\begin{array}{l}\displaystyle
\parallel v^j - v^{j - 1}\parallel _{X_T^s } = \parallel \int_0^t
W(t - \tau )(\frac{1}{2}({\left[ {(v^{j - 1})}^2 \right]_x -
\left[ {(v^{j - 2})}^2 \right]_x })+\frac{1}{3}({\left[ {(v^{j - 1})}^3 \right]_x -
\left[ {(v^{j - 2})}^3 \right]_x }))
d\tau\parallel _{X_T^s }\vspace{\smallskipamount}\\
\displaystyle~~~\leqslant \left( {3 + T} \right)^{1 / 2}\alpha _T^s
T^{1 / 2}\parallel \left[ {\left( {v^{j - 1}} \right)^2 + v^{j -
1}v^{j - 2} + \left( {v^{j - 2}} \right)^2} \right] \cdot \left[ {
v^{j - 1} - v^{j - 2} }
\right]\parallel _{X_T^s }\vspace{\smallskipamount}\\
\displaystyle~~~+\left( {3 + T} \right)^{1 / 2}\alpha _T^s
T^{1 / 2}\parallel \left[ v^{j - 1} +  v^{j - 2} \right] \cdot \left[ {
v^{j - 1} - v^{j - 2} }
\right]\parallel _{X_T^s }\vspace{\smallskipamount}\\
\displaystyle~~~\leqslant \left( {3 + T} \right)^{3/ 2}8\alpha
_T^s T^{1/2}\left\| \varphi\right\|_s^2 \parallel v^{j - 1} - v^{j
- 2}\parallel_{X_T^s }\vspace{\smallskipamount}+\left( {3 + T} \right)4\alpha _T^s
T^{1 / 2}\left\| \varphi\right\|_s \parallel v^{j - 1} - v^{j
- 2}\parallel_{X_T^s }\vspace{\smallskipamount}\\
\displaystyle~~~\leqslant \frac{1}{2}\parallel
{v^{j-1}-v^{j-2}}\parallel
_{X_T^s}.\end{array}\end{equation*}

If  $\alpha _T^s T^{1/2}[8(3 + T)^{3/2}\parallel\varphi
\parallel _s^2+4(3+T)\parallel
\varphi\parallel_s]\leqslant\frac{1}{2},$ then we obtain the result that
 \begin{equation*}\parallel v^{j + 1}\left(t\right) -
v^j\left( t \right)\parallel _s \leqslant
\parallel v^{j + 1}\left( t \right) - v^j\left(t\right)\parallel_{X_T^s }
\le 2^{ - j-1 }\left( {3 + T} \right)^{1 / 2}\parallel \varphi
\parallel _s.\end{equation*}
\end{proof}
\begin{proposition}\label{prop-3.4}Let
$v\left(t\right) =\bar{S}\left( {v, \varphi}\right)\left( t \right), ~v_n
\left(t\right)=\bar{S}\left({v_n ,\varphi_n }\right)\left( t \right). $
If  \begin{equation*}\begin{array}{l}\displaystyle~~~\alpha _T^s T^{1/2}(3 + T)^{3/2}(12\parallel\varphi
\parallel _s^2+16\parallel
\varphi\parallel_s+6)\leqslant\frac{1}{2},\end{array}\end{equation*}  then we have
\begin{equation*}
\left\|{v\left(t\right)-v_n\left(t\right)}\right\|_s\leqslant 2\left({3+
T}\right)^{1/2}\left\|{\varphi-\varphi_n}\right\|_s,\end{equation*}
where $\alpha_T^s=\sqrt
s\cdot2^s\cdot T^{1/2}+1~.$\\
\end{proposition}

\begin{proof}
Since $v\left(t\right) =\bar{S}\left( {v, \varphi}\right)\left( t
\right), ~v_n \left(t\right)=\bar{S}\left({v_n ,\varphi_n
}\right)\left( t \right), $~by Lemma 4.8 in \cite{dMS03}, we obtain the result
as following:
\begin{equation*}\begin{array}{l}\displaystyle
 \parallel v-v_n\parallel_{X_T^s }=\parallel W\left(t\right)\left(
{\varphi -\varphi_n }\right)-\frac{1}{2}\int_0^t{W\left({t-\tau}
\right)}\left[{v^2-v_n^2}\right]_x d\tau-\frac{1}{3}\int_0^t{W\left({t-\tau}
\right)}\left[{v^3-v_n^3}\right]_x d\tau\parallel_{X_T^s}
\vspace{\smallskipamount}\\ \displaystyle~\leqslant
\left({3+T}\right)^{1/2}\parallel \varphi-\varphi_n
\parallel _s + \left( {3 + T} \right)^{1 / 2}\alpha_T^s T^{1/
2}\parallel v^2+vv_n+v_n^2
\parallel_{X_T^s}\cdot\parallel v-v_n \parallel_{X_T^s}\\+ \left( {3 + T} \right)^{1 / 2}\alpha_T^s T^{1/
2}\parallel v+v_n\parallel_{X_T^s}\cdot\parallel v-v_n \parallel_{X_T^s}.
\end{array}\end{equation*}
By Proposition 2.2, if $\alpha _T^s T^{1/2}(3 + T)^{3/2}(12\parallel\varphi
\parallel _s^2+16\parallel
\varphi\parallel_s+6)\leqslant\frac{1}{2},$notice that
$\parallel\varphi_n\parallel_s\leqslant
\parallel\varphi\parallel_s+1, $ then
\begin{equation*}\begin{array}{l}
~~\parallel v -v_n\parallel_{X_T^s}
\leqslant \left({3 + T}\right)^{1/ 2}\parallel \varphi - \varphi_n
\parallel _s + 12\left( {3 + T} \right)^{3 / 2}\alpha _T^s T^{1 /
2}\parallel\varphi \parallel _s^2\cdot\parallel
v-v_n\parallel_{X_T^s}\\~~~~~~~~~~~~~~~~~~~+4\left( {3 + T} \right)\alpha _T^s T^{1 /
2}\parallel\varphi \parallel _s\cdot\parallel
v-v_n\parallel_{X_T^s}\\
~~~~~~~~~~~~~~~~~~~\leqslant\left({3+ T}\right)^{1 / 2}\parallel \varphi
-\varphi_n\parallel_s+\frac{1}{2}\parallel v
-v_n\parallel_{X_T^s}.\end{array}\end{equation*} Therefore $\parallel v-v_n
\parallel_{X_T^s }\leqslant 2\left({3+T}\right)^{1/
2}\parallel\varphi-\varphi_n\parallel_s,$ the sequence $\{v_n\}$ is
uniform convergence.
\end{proof}

Finally, we give the main results are as follows:
\begin{theorem}\label{thm-3.5}
The solution operator $K_\IR: H^s(\IR) \to C\left(\IR;H^s(\IR)\right)$  of the
initial value problem problem {\rm(\ref{equ-3})} is $\left({\delta_{H^s},
\left[{\rho\to\delta_{H^s}}\right]}\right)$-computable for any integer
$s\geqslant 3$.
\end{theorem}

\begin{proof}
For a given initial value $\varphi\in H^s (\mathbb{R})$ and a
rational number $\bar{T}>0$ we will show how to compute the solution
$u (t)$ of the initial value problem (\ref{equ-3}) at the time
interval $0\leqslant t \leqslant\bar{T}$. For this purpose,  we
first find some appropriate rational number $T$ such that
$0<T<\bar{T}$,  and show how to compute $u(t)$ from $t^\prime $ and
$\psi:=u({t}^\prime)$ at the time interval $\left[{{t}^\prime,
{t}^\prime+T}\right]$,  $0\leqslant t^\prime \leqslant\bar{T}$,  by
a fixed point iteration. Using this method, we can compute the values
$u\left({T /2m}\right)$ successively for $m=1, 2, \cdots$ and
finally $u\left(t\right)$ for any $0\leqslant t \leqslant\bar{T}$.

If $u_t+uu_x+u^2u_x+u_{xxx}=0, u\left({x, {t}^\prime }\right)=\psi\left(x\right)$,
and $v$ is defined by $v\left({x, t}\right):=u\left({x, t + {t}^\prime}\right),
$ then
\begin{eqnarray}\label{equ-7}
\left\{\begin{array}{l}
v_t+vv_x+v^2v_x+v_{xxx}  =0 \quad x\in{\mathbb{R}}, t\geqslant 0, \\
v({x, 0}) =\psi(x).
\end{array}\right.
\end{eqnarray}
We assume that the initial value $\psi\in H^s(\IR)$ is given by a
$\tilde{\delta}_{H^s}$-name,  i.e.,  by a sequence $\psi_{0},
\psi_1,  \cdots $ of Schwartz functions such that
$\parallel\psi-\psi _n\parallel_s\leqslant 2^{-n}$. For any
$n\in{\mathbb{N}}$,  we define function $v_n^0, v_n^1, \cdots$  in $C
( {{\mathbb{R}}:S ({\mathbb{R}} )} )$ by
\begin{equation*}
v_n^0:=\bar{S}({0, \psi_n}), \quad v_n^{j+1}:=\bar {S}({v_n^j, \psi_n}).
\end{equation*}
We note that the sequence $\{v_n^j\}$ can be computed from $\psi_n$. By Proposition 3.3, the
iterative sequence $v_n^0, v_n^1, \cdots $ converges to some $v_n$,  then $v_n$
is the fixed point of the iteration $\bar{S}$ and satisfies the following
internal equation:
\begin{equation*}
  v_n\left(t\right)=\bar{S}\left({v_{n, }\psi_n}\right)= \mathcal{F}^{-1}({E(t)\cdot \mathcal{F}(\varphi )})
 -\int_0^t{\mathcal
{F}^{-1}\left({E\left({t-\tau}\right)\cdot \mathcal
{F}\left({\textstyle{d\over{dx}}\left({1\over2}{u^2\left(\tau\right)}+{1\over3}{u^3\left(\tau\right)}
\right)}\right)}\right)}\mathrm{d}\tau
\end{equation*}
hence $v_n$ solves the initial value problem:
\begin{equation*}
\frac{\partial v_n}{\partial t}+v_n\frac{\partial v_n}{\partial x} +v_n^2\frac{\partial v_n}{\partial x} +
\frac{\partial v_n}{\partial x^3}=0, \quad \mbox{ and } \quad v_n (x, 0)=
\psi_n(x).
\end{equation*}
By Proposition 2.3, we will show that,  by a contraction argument,  for some
sufficiently small computable real number $T>0$ (depending only on
$\varphi$ and $\bar{T}$), $v_n^j\left(t \right)\to
v_n\left(t\right)$ as $j\to \infty$ for all $n$,  and $v_n \left(t
\right)\to v\left(t\right)$ as $n\to \infty, $ sufficiently fast and
uniformly in $t \in \left[{0, T}\right]$. We recall that $v$ is the
solution of the initial value problem (\ref{equ-5}). Then we can effectively
determine a computable subsequence of the double sequence$\{v_n^j\}$
which will converge fast to $v$ uniformly in $t \in \left[{0,
T}\right]$.

Since $v$ is the limit of a fast convergent computable sequence, $v$
itself is computable. So $K_{\mathbb{R}}: ({\varphi, t})\mapsto
u\left(t\right)$ is $\left({\delta _{H^s},  \rho,
\delta_{H^s}}\right)$-computable for $t\geqslant 0$,  the reflection
$ R:S\left({\mathbb{R}}\right)\to S ({\mathbb{R}}), R(\psi ) (x):=
\psi (- x)$,  is $({\delta_{H^s}, \delta_{H^s}})$-computable. Define
${u}^\prime \left(t\right)\left(x\right):=u\left({-t}\right)\left({-
x}\right)$. Then ${u}^\prime _t+{u}^\prime{u}^\prime _x+{u}^{\prime{2}}{u}^\prime _x
+{u}^\prime _{xxx}=0$ and for $t\geqslant 0, u\left({-t}\right)=R
\circ{u}^\prime \left(t\right)=R \circ K_\mathbb{R}\left({{u}^\prime
\left(0\right), t}\right)=R \circ
K_\mathbb{R}\left({R\left(\varphi\right), t}\right)$,  i.e.,
$u\left(t\right)=R \circ K_\mathbb{R}\left({R\left(\varphi\right),
-t}\right)$ for $t\leqslant 0$. Therefore,  as the two computable
functions join at $0$, $K_\mathbb{R}$ is computable for
$t\in\mathbb{R}$. (see  \cite{dMS03})
\end{proof}

Thus, we prove the main result, and we can see that the machine searches for fast approximations to
$u(x; t)$, and computes the solutions of the Combined KdV equation with arbitrary precision. This approach
 can be extended to other nonlinear equations.

\end{document}